\documentclass[12pt]{article}
\usepackage[utf8]{inputenc}
\usepackage[T1]{fontenc}
\usepackage{mathptmx}
\usepackage{bm}
\usepackage{latexsym}
\usepackage{pstricks}
\usepackage{graphicx}
\usepackage{amssymb}
\usepackage{amsmath}
\usepackage{xcolor}
\usepackage{soul}
\usepackage{cite}
\usepackage{caption}

\usepackage{amsmath}
\usepackage{amsfonts}
\usepackage{amsthm,array} 
\usepackage{array} 
\usepackage{graphicx}
\usepackage{txfonts}
\usepackage{pxfonts}

\usepackage{float}  
\usepackage{longtable}
\usepackage{arydshln}

\usepackage[normalem]{ulem}

\usepackage{mathtools}

\usepackage{array}
\usepackage{subcaption}
\usepackage{booktabs}

\def\bR{{\mathbb{R}}} 

\def\bz{{\mathbb{Z}}}

\makeatletter

\@addtoreset{equation}{section}
\makeatletter

\setstcolor{red}

\def\bR{{\mathbb{R}}} 

\def\bz{{\mathbb{Z}}}

\def\vs{\vskip.3cm}
\def\noi{\noindent}
\def\q{\quad}

\usepackage{color}
\newcommand{\rc}{\textcolor{red}}

\newtheorem{thm}{Theorem}[section]
\newtheorem{lem}[thm]{Lemma}
\newtheorem{cor}[thm]{Corollary}
\newtheorem{proposition}[thm]{Proposition}
\newtheorem{Def}[thm]{Definition}
\newtheorem{rmk}[thm]{Remark}

\newtheorem{ex}[thm]{Example}

\newcommand{\END}{\hfill\mbox{\raggedright$\Diamond$}}

\def\vs{\vskip.3cm}

\def\la{\langle}
\def\ra{\rangle}

\def\Ind{\text{\rm Ind\,}}

\def\sign{\text{\,sign\,}}

\def\rind{\text{\,rind\,}}

\def\const{\text{\,constant\,}}

\usepackage{ulem}

\begin{document}

\vfill
\begin{flushright}
\end{flushright}

\begin{center}
\baselineskip=16pt
{\Large\bf Degeneracy Index and Poincar\'e-Hopf Theorem}
\vskip 0.5cm
\vskip 10.mm

{\bf Haibo Ruan${}^{a}$ and Jorge Zanelli${}^{b}$} \\
\vskip 1cm
{${}^a$ Department of Mathematics, University of Hamburg, Bundesstrasse 55, D-20146 Germany\\
${}^b$ Centro de Estudios Cient\'{\i}ficos (CECs), Av. Arturo Prat 514, Valdivia, Chile \\
\vskip 0.2cm
\texttt{\footnotesize{haibo.ruan@math.uni-hamburg.de, z@cecs.cl}}} 
\vspace{6pt}

\begin{abstract}
A degenerate dynamical system is characterized by a state-dependent multiplier of the time derivative of the state in the time evolution equation. It can give rise to Hamiltonian systems whose symplectic structure possesses a non-constant rank throughout the phase space. Around points where the multiplier becomes singular, flow can experience abrupt and irreversible changes. We introduce a topological index for degenerate dynamical systems around these {\it degeneracy points} and show that it refines and extends the usual topological index in accordance with the Poincar\'e-Hopf Theorem. 

\end{abstract} 
\end{center}

\section{Introduction} 

It can occur in some physical systems that they evolve into a state for which the coefficient of the highest derivative in the differential equation that governs the evolution of the system vanishes. When such state is reached the dynamical evolution experiences an abrupt change, the evolution may become unpredictable, some degrees of freedom may cease to exist, and information about the initial state can be lost \cite{Dds}. This may happen, for example, in compressible fluids when a wave front exceeds the speed of sound, as well as in shock-wave solutions of Burgers? equation \cite{Choodnovsky}. In Lagrangian mechanics this corresponds to a globally non-invertible relation between velocities and momenta, which leads to multiple Hamiltonians for a given Lagrangian \cite{HTZ} and has recently been used  to  construct ``time crystals" \cite{Wilczek}. This problem is also a generic condition of gravitation theories in more than four spacetime dimensions \cite{TZ}.

Although in mechanics the problem of degeneracy usually appears as the multivaluedness of the velocity as a function of momentum, in its simplest form degeneracy appears in autonomous first order equations such as those describing a Hamiltonian flow.
\vskip 0.3cm

Recall the usual form of continuous dynamical systems
\begin{equation}
\frac{dx}{dt}=f(x),\q x\in \bR^n, \q t\in \bR.\label{eq:1}
\end{equation}
The {\it degenerate dynamical systems} have a modified form of (cf. \cite{Dds})
\begin{equation}
 A(x)\cdot \frac{dx}{dt}=f(x),\label{eq:2}
\end{equation}
where the matrix $A$ on the left-hand side can become singular on {certain {\it degeneracy set}, defined by}
\begin{equation*}
{D} = \{x\in \bR^n\,:\, \det A(x)=0\}.
\end{equation*}

A characteristic feature of degeneracy is that around degeneracy points, magnitude of the velocity $|\frac{dx}{dt}|$ can be infinitely large. This can be seen from the following $1$D example 
\[
x\cdot \frac{dx}{dt}=1,\q x\in \bR,
\]
which degenerates at $x=0$. This degeneracy point marks the sign-change of $\det A(x)=x$, which results in the velocity $\frac{dx}{dt}$ switching from $+\infty$ to $-\infty$ as $x$ moves from the positive to the negative. In general, for higher dimensional phase space ($n>1$), degeneracy sets ${D}$ typically form co-dimension-1 surfaces in phase space (cf. \cite{Dds}).

Of particular interest in physics is the case of Hamiltonian flow, where $n$ is even, $x$ are the coordinates in the phase space and $A$ is a pre-symplectic form. It can be seen that in these cases the orbits in phase space are restricted to non-overlapping regions of the entire phase space, there are no orbits connecting those different regions and therefore the dynamics splits into distinct regimes separated by co-dimension-$1$ surfaces \cite{Dds}. Preliminary results also indicate that some of these features are shared by the canonically quantized version of degenerate systems: the separation between phase space regions translates into a splitting of the Hilbert space into a direct sum of orthogonal Hilbert spaces constructed on the classically disconnected regions of phase space \cite{Qds}.

Mathematical studies of degenerate dynamical systems can be found in the context of differential-algebraic equations (DAEs). Equations of form (\ref{eq:1}) correspond to linearly implicit differential equations and points of degeneracy are singular points of these equations (cf. \cite{RCM_2000}). Singular points in non-scalar case $n>1$ were first studied by Rabier in \cite{Rabier_1989}, followed by Medved \cite{Med_1991}, Rei\ss ig \cite{Rei_1997} where normal forms were obtained for certain types of singular points (cf. \cite{RB_1998}), who introduced a notion of {\it standard singular points}. Further studies led to normal forms for (\ref{eq:1}) near impasse points that are not necessarily standard. In piecewise smooth or discontinuous dynamical systems, it is generally assumed that the right-hand side $f$ of (\ref{eq:1}), defined on disjoint open domains of the phase space, can be continuously extended to their border (cf. \cite{Filippov_1988,Bernardo_2008,Hogan_2016}). On the other hand, degenerate dynamical systems {in which $\det A(x)$ has a simple zero are equivalent to those of the form} (\ref{eq:1}) whose right-hand side switches from $+\infty$ to $-\infty$ across the border of domain.

In multiple time scale dynamical systems, fast-slow systems are studied using singular perturbations, where dynamics on the slow manifold is described by a DAE which may or may not contain singular points (cf. \cite{Kuehn_2015, PS_2007}).

In this paper, we are interested in studying $2$D degenerate dynamical systems, where $A$  arises as a symplectic form
\[
A=\left(\begin{array}{cc} 0 &f\\ -f&0\end{array}\right),
\]
for a smooth function $f:\bR^2\to \bR$. The symplectic nature of $A$ is not a fundamental requirement for the present discussion, but it is motivated by physical considerations (cf. \cite{Dds}). 

That is, we are interested in studying 
\begin{equation}\label{eq:sym2d_gen}
\left(\begin{array}{cc} 0 & f\\ -f & 0\end{array}\right) 
\left(\begin{array}{c} \dot x_1\\\dot x_2 \end{array}\right)
=\left(\begin{array}{c} E_1\\ E_2 \end{array}\right)=E,
\end{equation}
for a smooth function $E:\bR^2\to \bR^2$. The degeneracy points in this case, are precisely zeros of $f$, which under a regularity assumption, form a co-dimension-1 submanifold in the plane, being either an infinitely extending line or a circle. 

The goal of this paper is to define a topological index {for the flow (\ref{eq:sym2d_gen}) in the presence of} these co-dimension one lines of degeneracy, which can be used to classify degenerate dynamics in the plane and on the two-dimensional sphere $S^2$. {Naturally}, degenerate dynamical systems with an empty degeneracy set become ordinary dynamical systems. As we will see, the introduced index gives an extension of the usual topological index for ordinary flows and provides a parallel of the Poincar\'e-Hopf {theorem} on the sphere for degenerate flows.

\section{Ring Index} \label{sec:ri}   

Let $f:\bR^2\to \bR$ be a smooth map having $0$ as a regular value. Assume that $f^{-1}(0)\ne \emptyset$. Then,  by the Implicit Function Theorem, the pre-image
\[\gamma:=\{(x_1,x_2)\in\bR^2\,:\, f(x_1,x_2)=0\}\]
is a submanifold of co-dimension one. Every connected component of $\gamma$ is either  homeomorphic to an infinitely extending line $\bR$ or to a circle $S^1$.

The reason {for requiring regularity of} $\gamma$ is to avoid singular cases such as self-intersections, which can be thought of as an intermediate transition between two topologically distinct sets. See Figure \ref{F:self} for example, where a figure-eight curve can be perturbed to either a disjoint union of two circles ($S^1\sqcup S^1$) or to one circle ($S^1$).
	\begin{figure}[!htb]
\centerline{
		\includegraphics[width=.45\textwidth]{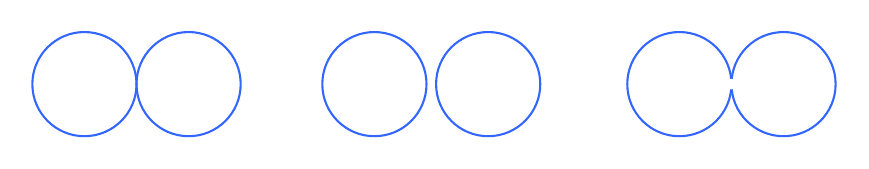}
} 

	\caption{A curve with self-intersection (left) can be either perturbed to a disjoint union of two circles  (middle) or  one circle  (right).}
	\label{F:self}  
\end{figure} 

Thus, no topological index that {remains} constant under small perturbations can be defined directly to such singular cases like the figure-eight. These can be studied in the context of topological bifurcations {and lies beyond} the scope of this paper.

We introduce first an index {for} compact $\gamma$ and then extend it to non-compact $\gamma$ using its compactification on the sphere $S^2$.

Denote by $J=\left(\begin{array}{cc} 0 & -1\\ 1 & 0\end{array}\right)$. Then, (\ref{eq:sym2d_gen}) can be reformulated as
\begin{equation}\label{eq:sym2d_gen_JE}
f
\left(\begin{array}{c} \dot x_1\\\dot x_2 \end{array}\right)
=JE.
\end{equation}

\subsection{Definition}\label{subsec:def_ri}
\begin{Def}\rm A compact connected $1$-manifold $\gamma$ without boundary is called a {\it ring}. If $\gamma=f^{-1}(0)$ for a smooth map $f:\bR^2\to \bR$, then $f$ is chosen (with an appropriate sign) so that $J\nabla f$ coincides with the counter-clock wise orientation on $\gamma$.

\END
\end{Def}

\begin{ex}\label{ex:f_2}\rm
\begin{itemize}
	\item[(a)]	Let $\gamma=\{(x_1,x_2)\in \bR^2: x_1^2+x_2^2=1\}$ and $f=x_1^2+x_2^2-1$. Then, $\gamma=f^{-1}(0)$ and $J\nabla f$ is counter-clock wise on $\gamma$. See Figure \ref{F:f_ex_2} (left). 
	\begin{figure}[!htb]
		\centerline{
			\includegraphics[width=.4\textwidth]{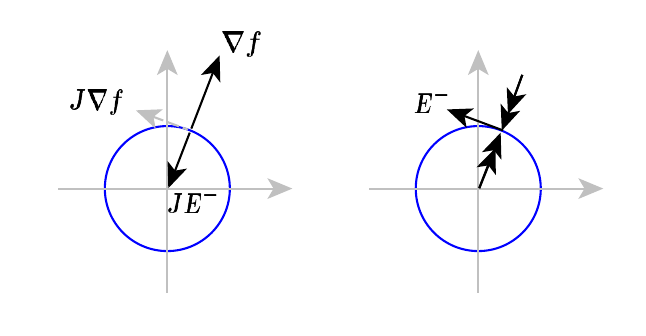}
		} 

		\caption{Left: Positions of $J\nabla f$ and $JE^-$ on $\gamma=\{(x_1,x_2)\in \bR^2: x_1^2+x_2^2=1\}$, where $f=x_1^2+x_2^2-1$ and $E^-=(E_1,E_2)=(-x_2,x_1)$. Right: Dynamics of (\ref{eq:sym2d_gen_JE}) around $\gamma$. }
		\label{F:f_ex_2}  
	\end{figure} 
 
\item [{(a)}] Consider $E^-:=(E_1,E_2)=(-x_2,x_1)$ for $(x_1,x_2)\in \bR^2$. Then, $JE^-=(-x_1, -x_2)$. The dynamics of (\ref{eq:sym2d_gen_JE}) around $\gamma$ can be obtained directly. It is characterized by an in-flow from both sides of $\gamma$ with an infinite velocity $(\dot x_1,\dot x_2)$ along $JE^-$. See Figure \ref{F:f_ex_2} (right). 
\item [(b)] Consider $E^+=(E_1,E_2)=(x_2,-x_1)$ for $(x_1,x_2)\in \bR^2$. Then, $JE^+=(x_1, x_2)$. The dynamics of (\ref{eq:sym2d_gen_JE}) around $\gamma$ can be obtained similarly. It is characterized by an out-flow from both sides of $\gamma$ with an infinite velocity $(\dot x_1,\dot x_2)$ {along $JE^+$}. See Figure \ref{F:f_ex_3} (right).  
 
 	\begin{figure}[!htb]
 	\centerline{
		\includegraphics[width=.4\textwidth]{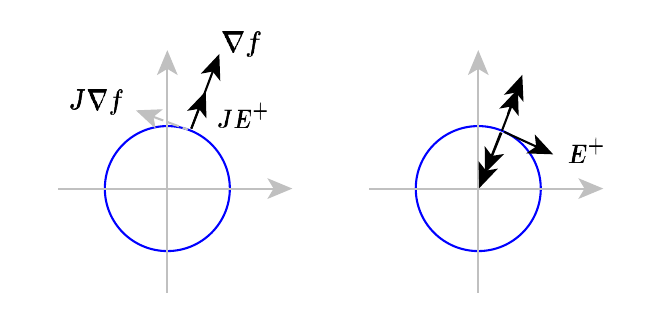}
	} 
 	
 	\caption{Left: Positions of $J\nabla f$ and $JE^+$ on $\gamma=\{(x_1,x_2)\in \bR^2: x_1^2+x_2^2=1\}$, where $f=x_1^2+x_2^2-1$ and $E^+=(E_1,E_2)=(x_2,-x_1)$. Right: Dynamics of (\ref{eq:sym2d_gen_JE}) around $\gamma$.}
 	\label{F:f_ex_3}  
 \end{figure} 

\end{itemize}
\END
\end{ex}	

\begin{rmk}\rm
	Note that the two vector fields $JE^\pm$ in Example \ref{ex:f_2}(a)-(b) are homotopic on $\gamma$ by a rotation of $180^\circ$. Indeed, they both define a map $\gamma\to S^1$ that has winding number $1$. However, since the degenerate dynamics they define through (\ref{eq:sym2d_gen_JE}) differs dramatically, we want to  introduce a topological index that can distinguish the two cases, by putting restrictions on the allowed deformations.
\END
\end{rmk}

The idea is to refrain $JE$ from taking the tangent direction of $\gamma$. For convenience, denote  by
$$m_\gamma(JE)=\big|\{(x_1,x_2)\in \gamma: \la JE(x_1,x_2), \nabla f (x_1,x_2)\ra=0\}\big|,$$
which is the number of points on $\gamma$ where $JE$ becomes tangent to $\gamma$ {and $\la \cdot , \cdot \ra$ is the Euclidean scalar product in $\bR$}.  In case $\gamma$ becomes a point as a limiting case of a degeneracy ring, we set $m_\gamma(JE)=0$. This number is by no means always finite. An extreme example is when $JE$ becomes tangent everywhere on $\gamma$, a case of {\it reducible degeneracy} (cf. \cite{Qds,Dds}). In example 2.2, $\la JE^{{\pm}} , \nabla f \ra=\pm 2$ and therefore $m_\gamma(JE)=0$.

\begin{lem}\label{lem:even} Consider the map $\la JE, \nabla f\ra:\gamma\to \bR$, which is a differentiable map between manifolds of the same dimension. If $0$ is a regular value of $\la JE, \nabla f\ra$, then there are an even number of zeros of $\la JE, \nabla f\ra$ on $\gamma$.
\end{lem}

\begin{proof} Consider the map $\la JE, \nabla f\ra$ as a map $g:[a,b]\to \bR$ such that $g(a)=g(b)$ {Here we used the {implicit assumption} that $\gamma$ is a closed curve, $\gamma \sim S^1$}. Then, the graph of $g$ must cross even number of times on the $x$-axis to come back the initial value $g(a)=g(b)$.  
	
Formally, $\bR$ is contractible, so any continuous map from $S^1\to\bR$ is homotopic to a (non-zero) constant map, whose number of zeros is zero. By the modulo $2$ degree of maps between manifolds of the same dimension, homotopic maps have the same number of zeros modulo $2$. Thus, $\la JE, \nabla f\ra$ must have an even number of zeros on $S^1$. See \cite{Milnor}. 
\end{proof}

\begin{Def}\rm
A homotopy $H:[0,1]\times \gamma\to \bR^2\setminus\{0\}$ on vector fields $E$ is called {\it admissible}, if $m_{\gamma}(JH(t,\cdot))\equiv\const$ for all $t\in[0,1]$. Two vector fields $E_0,E_1$ are called {\it admissibly homotopic}, if there is an admissible homotopy $H$ connecting the two: $E_i=H(i,\cdot)$ for $i=0,1$.

Similarly, a homotopy $H:[0,1]\times S^1\to \bR^2$ on rings $\gamma$ is called {\it admissible}, if $m_{H(t,\cdot)} (JE)\equiv\const$ for all $t\in[0,1]$.
	\END
\end{Def}

\begin{Def}\rm\label{def:ring_ind} Let $\gamma=f^{-1}(0)\subset\bR^2$ be a ring and $E:\gamma\to \bR^2\setminus\{0\}$ be a vector field such that $0$ is a regular value of $\la JE,\nabla f\ra$. Let $m_\gamma(JE)$ be the (even) number of zeros of $\la JE,\nabla f\ra$.
	\begin{itemize}
		\item [(a)] If $m_\gamma(JE)=0$, then $\la JE, \nabla f\ra\ne 0$ for all points on $\gamma$. Define the {\it ring index} of $\gamma$ by
		\[\rind_{\gamma}(E):=\frac{1}{l(\gamma)}\underset{s\in \gamma}{\oint}\sign \la JE, \nabla f\ra ds,\]
		where $l(\gamma)$ is the length of $\gamma$ and {$\sign (c)=\pm 1$ for $\pm c>0$.} 
		
	\item[(b)]	If $m_\gamma(JE)>0$, then $\gamma$ is divided {into $m:=m_\gamma(JE)$ intervals $\gamma_1,\dots, \gamma_m$, between consecutive zeros}. Define the {\it ring index} of $\gamma$ by
	\[\rind_\gamma(E)=\overset{m}{\underset{i=1}{\sum}} \rind_{\gamma_i} (E).\]
	\end{itemize}
	\END
\end{Def}	

\begin{ex}\rm\label{ex:rind_E_pm}
Let $E^\pm$ be the two vector fields given in Example \ref{ex:f_2}(a)-(b). Then, we have $m_\gamma(JE^\pm)=0$ and $\rind_{\gamma}(E^\pm)=\pm 1$.
\END
\end{ex}

\begin{rmk}\rm
It is interesting to notice that
\[\la JE,\nabla f\ra=\la E, -J\nabla f\ra\]
by the ortho-normal property $-J^2=1$. Thus, the ring index is also equal to
\begin{equation*}
\rind_\gamma(E)=- \frac{1}{l(\gamma)} \oint_\gamma \sign \la E, J\nabla f \ra ds.
\end{equation*}
If $m_\gamma(JE)=0$, then $\la E, J\nabla f \ra$ does not change sign on $\gamma$. Thus, we have
\[\rind_\gamma(E)=- \sign\big (\frac{1}{l(\gamma)} \oint_\gamma \la E, J\nabla f \ra ds\big).
\]
Since $J\nabla f$ is the tangent vector in alignment of the orientation of $\gamma$, the ring index expresses the work done {against $E$ along $\gamma$}.
\END
\end{rmk}

\begin{lem}\label{lem:r1} The ring index satisfies the following.
	\begin{itemize}
		\item [(a)] If $m_\gamma(JE)=0$, then $\rind_\gamma(E) =\sign \la JE, \nabla f\ra\in\{\pm 1\}$.
		\item[(b)] If $m_\gamma(JE)>0$, then $\rind_\gamma(E) =0$.
	\end{itemize} 
\end{lem}
\begin{proof} (a) Since $m_\gamma(JE)=0$, $\la JE, \nabla f\ra$ vanishes nowhere on $\gamma$ and $\sign\la JE, \nabla f\ra \equiv \pm 1$ is a constant function. Thus,
	\[\rind_\gamma(E)=\sign\la JE, \nabla f\ra\frac{1}{l(\gamma)}\int_{s\in \gamma}ds=\sign\la JE, \nabla f\ra.\]

\noi(b) By Lemma \ref{lem:even}, $m_\gamma(JE)$ is an even number. Thus, $\gamma$ is divided into even number of  intervals, on each of which $\rind_{\gamma_i}(E)\in\{\pm 1\}$ has an alternating sign. Therefore,  $\rind_{\gamma}(E)=0$.

\end{proof}

\begin{proposition}\label{prop:ring_ind} The ring index is a homotopy invariant under all admissible homotopies of vector fields $E$ and of rings $\gamma$.		
\end{proposition}
\begin{proof}
It follows from Lemma \ref{lem:r1}, since admissible homotopies keep the number $m_\gamma(JE)$ constant, under which condition the value of $\rind_\gamma(E)$ remains the same.	

\end{proof}

\begin{rmk}\rm \label{rmk_rev_rind}
The reverse of Proposition \ref{prop:ring_ind} does not hold in general. {That is, two vector fields  having the same ring index are {\it not} necessarily admissibly homotopic.} Indeed, if $E_1,E_2$ are vector fields with $m_\gamma(JE_i)=2i>0$. Then, $\rind_\gamma(E_i)=0$ for $i=1,2$. However, the two vector fields are not admissibly homotopic to each other, due to their different numbers $m_\gamma(JE_1)\ne m_\gamma(JE_2)$.
\END
\end{rmk}

It turns out that the reverse statement of Proposition \ref{prop:ring_ind} holds for vector fields with $m_\gamma(JE)=0$.

\begin{proposition}\label{prop:ring_ind_2}  Let $E^\pm$ be the two vector fields given in Example \ref{ex:f_2}(a)-(b). Then,  for any vector field $E$ with $m_\gamma(JE)=0$, we have
	\begin{itemize}
		\item[(i)] $\rind_\gamma(E)=-1$ if and only if $E$ is admissibly homotopic to $E^-$;
		\item[(ii)] $\rind_\gamma(E)=1$ if and only if $E$ is admissibly homotopic to $E^+$.
	\end{itemize}
	
\end{proposition}
\begin{proof} (i) If $\rind_\gamma(E)=-1$, then by Lemma \ref{lem:r1}, $\la JE,\nabla f\ra<0$ on the whole $\gamma$. Thus, the map $\la JE,\nabla f\ra<0$ can be continuously deformed to the constant map $-1$ on $\gamma$, by an admissible homotopy. Since $\la JE^-, \nabla f\ra \equiv -1$, this homotopy leads to an admissible homotopy from $E$ to $E^-$. Conversely, if $E$ is admissibly homotopic to $E^-$, then at every $t\in [0,1]$, we have $\sign \la JH(t,\cdot), \nabla f\ra=-1$. Thus, by Lemma \ref{lem:r1}(i), $\rind_\gamma(H(t,\cdot))=-1$ for  all $t\in [0,1]$. Especially, we have $\rind_\gamma(E)=-1$. 
	
	The part (ii) is parallel.
\end{proof}

\begin{Def}\rm
	A vector field $E:\gamma \to \bR^2\setminus\{0\}$ is called {\it rotating}, if $m_\gamma(JE)=0$.
	\END
\end{Def}

By Proposition \ref{prop:ring_ind_2}, the ring index gives a topological classification of rotating fields under admissible homotopies.

\subsection{Relation to the Winding Number}\label{subsec_rel}

Consider $E:\gamma\to \bR^2\setminus\{0\}$. Recall that the {\it winding number}  $w(E,\gamma)$ of $E$ along $\gamma$ counts how many times the image $E(\gamma)\subset \bR^2\setminus\{0\}$ has gone around the origin counterclockwise.

\begin{ex}\rm\label{ex:wind_E_pm}
	Let $E^\pm$ be the two vector fields given in Example \ref{ex:f_2}(a)-(b). Then, the winding number of $E^\pm$ is in both cases,  $w(E^\pm, \gamma)=1$. To the contrary, the ring index distinguishes the two by being  $\rind_\gamma(E^\pm)=\pm 1$ (cf. Example \ref{ex:rind_E_pm}).
	\END
\end{ex}

\begin{lem}\label{lem:rw} For rotating fields $E$,
we have $\rind_\gamma(E) =\sign \la JE, \nabla f\ra\cdot w(E,\gamma)$.
\end{lem}
\begin{proof} If $\la JE, \nabla f\ra\ne 0$ for all $(x_1,x_2)\in \gamma$, then $\la JE, \nabla f\ra <0$ or $>0$ always. Assume the first case.  By the homotopy invariance of the winding number, we have
\[w(E,\gamma)=w(\tilde E,\gamma),\]
where $\tilde E$ can be chosen to be a map such that $\la JE, \nabla f\ra \equiv -1$ on the whole $\gamma$. One such choice is  $\tilde E=E^-$. Thus, by  Proposition \ref{prop:ring_ind_2},
\[\rind_\gamma(E) =\rind_\gamma(E^-)=-1=-w(E^-,\gamma)=-w(E,\gamma).\]
The other case $\la JE, \nabla f\ra >0$ is similar.
 \end{proof}

\begin{cor}\label{cor:w1} For rotating fields $E$, we have $w(E,\gamma)=1$.
\end{cor}
\begin{proof} It is a consequence of Corollary \ref{lem:r1} and Lemma \ref{lem:rw}.
\end{proof}

\begin{cor}\label{lem:rw_ne_1} Every vector field with $w(E,\gamma)\ne 1$ has a ring index zero.

\end{cor}
\begin{proof}
Let $E$ be a vector field with $w(E,\gamma)\ne 1$. By Corollary \ref{cor:w1}, $E$ is not a rotating field and thus,  $m_\gamma(JE)>0$. By Lemma \ref{lem:r1}(b), $\rind_\gamma(E)=0$.
\end{proof}

\subsection{Robustness of Degeneracy Rings}\label{subsec:robust}

By Remark \ref{rmk_rev_rind} and Corollary \ref{lem:rw_ne_1}, the ring index does not distinguish among vector fields having winding numbers other than $1$. These include vector fields of winding number $0$ (constant field), $-1$ (saddle) or $2$ (dipole).

We will show that vector fields with winding numbers other than $1$ cannot sustain a stable ring of degeneracy on a simply connected region. Thus, the only rings of degeneracy that exist robustly in $\bR^2$ or $S^2$ are those with non-zero ring index and supported by rotating fields.

\begin{ex}\rm \label{ex:E_pm_n}
Consider the representative vector fields for winding numbers $\pm n$ with $n=0,1,\dots$. For a parametrization $(\cos t, \sin t)$ of $\gamma=S^1$ with $t\in [0,2\pi]$, the vector fields $E_n=(-\sin(nt), \cos(nt))$ and $E_{-n}=(\sin(nt),\cos(nt))$ have winding numbers $w(E_n)=n$ and $w(E_{-n})=-n$, respectively. Moreover, 
\begin{align*}
\la JE_n, \nabla f \ra&=-\cos(n-1) t\\
\la JE_{-n}, \nabla f \ra&=-\cos(n+1) t,
\end{align*}
where $f$ is chosen so that $J\nabla f=(-\sin t,\cos t)$ coincides with the tangent direction of $S^1$ in the counter-clock wise sense.
Thus,  the vector fields $E_{n+2}$ and $E_{-n}$ have the same number of zeros of $\la JE,\nabla f\ra$ on $S^1$, being equal to
\[m_\gamma(JE_{n+2})=m_\gamma(JE_{-n})=2(n+1).\]

{Conversely}, one can show that a vector field $E$ with $m_\gamma(JE)=2(n+1)$ is homotopic to either  $E_{n+2}$ or $E_{-n}$ by properties of winding numbers.
\END
\end{ex}

\begin{ex}\rm
	Consider a vector field $E$ with $m_\gamma(JE)=2$ for $\gamma=S^1$ and $f$ being chosen so that $J\nabla f$ points at the tangent of $S^1$ in the counter-clock wise direction. The degenerate dynamics of (\ref{eq:sym2d_gen_JE}) is depicted in Figure \ref{F:2_zeros}.
	\begin{figure}[!htb]
	\centerline{
		\includegraphics[width=.25\textwidth]{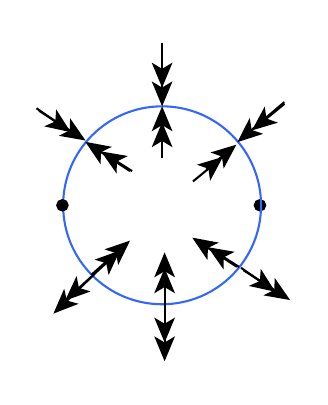}
	} 
\vskip0cm
	\caption{Dynamics of (\ref{eq:sym2d_gen_JE}) for a vector field $E$ having $m_\gamma(JE)=2$ on  $\gamma=S^1$.}
	\label{F:2_zeros}  
\end{figure} 

The winding number of $E$ along $\gamma$ is either $0$ or $2$, being  homotopic to either the constant field $E_{0}$ or to the dipole $E_2$, respectively. See Figure \ref{F:2_zero_a}, where  $\tilde E_2$ is a variation of $E_2$ that has only simple zeros.

	\begin{figure}[!htb]
	\centerline{
		\includegraphics[width=.5\textwidth]{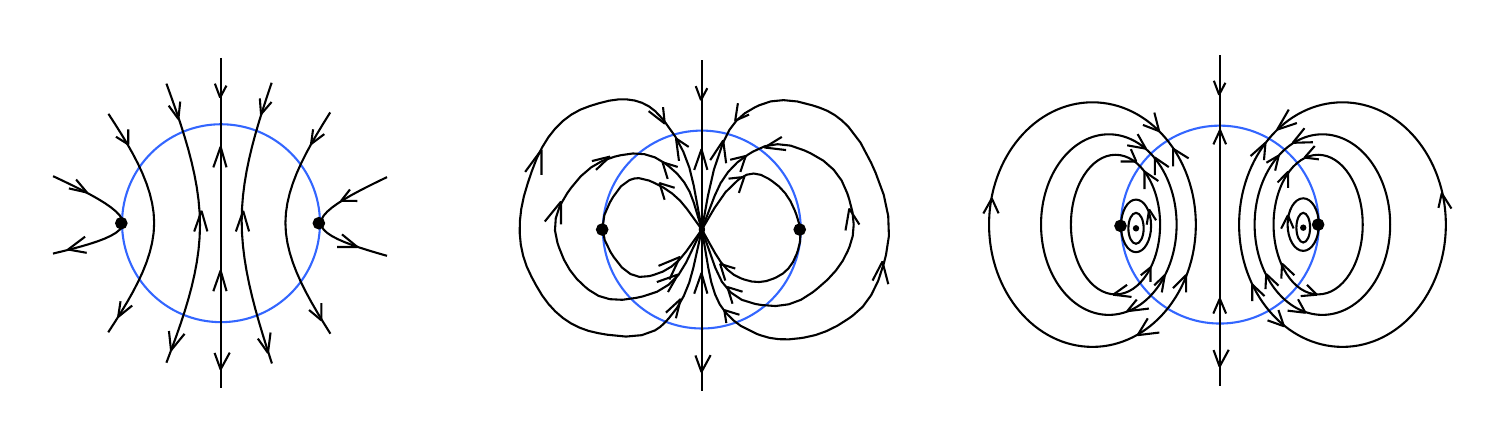}
	} 
	\vskip0cm
	\caption{Three kinds of vector fields $E_0$ (left), $E_2$ (middle) and $\tilde E_2$ (right) that have degenerate dynamics from Figure \ref{F:2_zeros}, where $\tilde E_2$ is homotopic to $E_2$ admissibly. All lead to unstable degeneracies. See Figure \ref{F:2_zero_b}.}
	\label{F:2_zero_a}  
\end{figure}  

The zeros on the ring in all cases {allow} the whole ring to {be continuously deformed collapsing it to a point.} The one side of in-flow finds its way out to the other side of out-flow. See figures \ref{F:2_zero_b}-\ref{F:2_zero_d}.

\begin{figure}[!htb]
	\centerline{		
		\includegraphics[width=.5\textwidth]{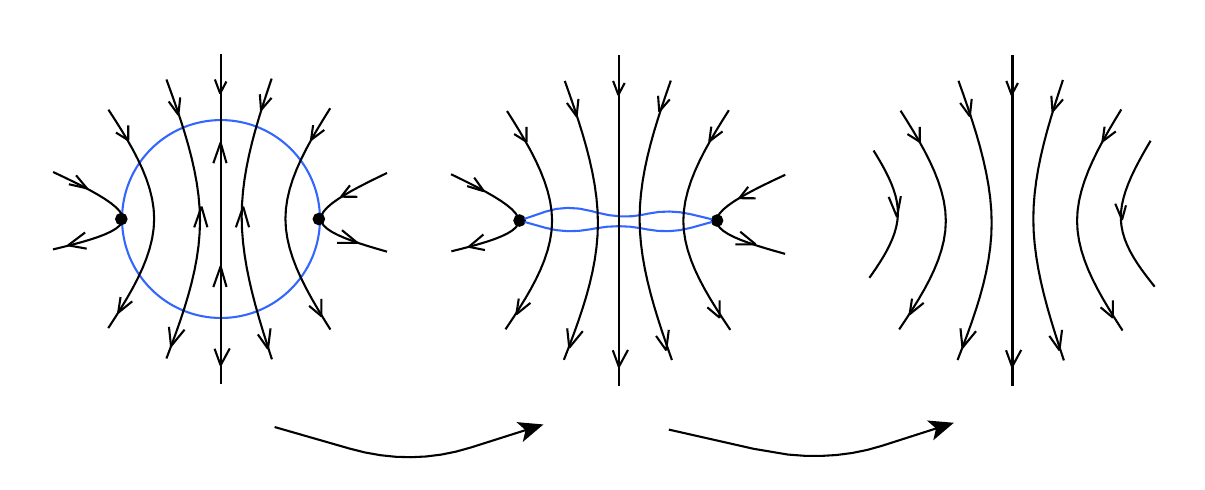}
	} 
	\vskip0cm
	\caption{The disappearance of degeneracy for $E_0$.}
	\label{F:2_zero_b}  
\end{figure}  

\begin{figure}[!htb]
	\centerline{				
		\includegraphics[width=.5\textwidth]{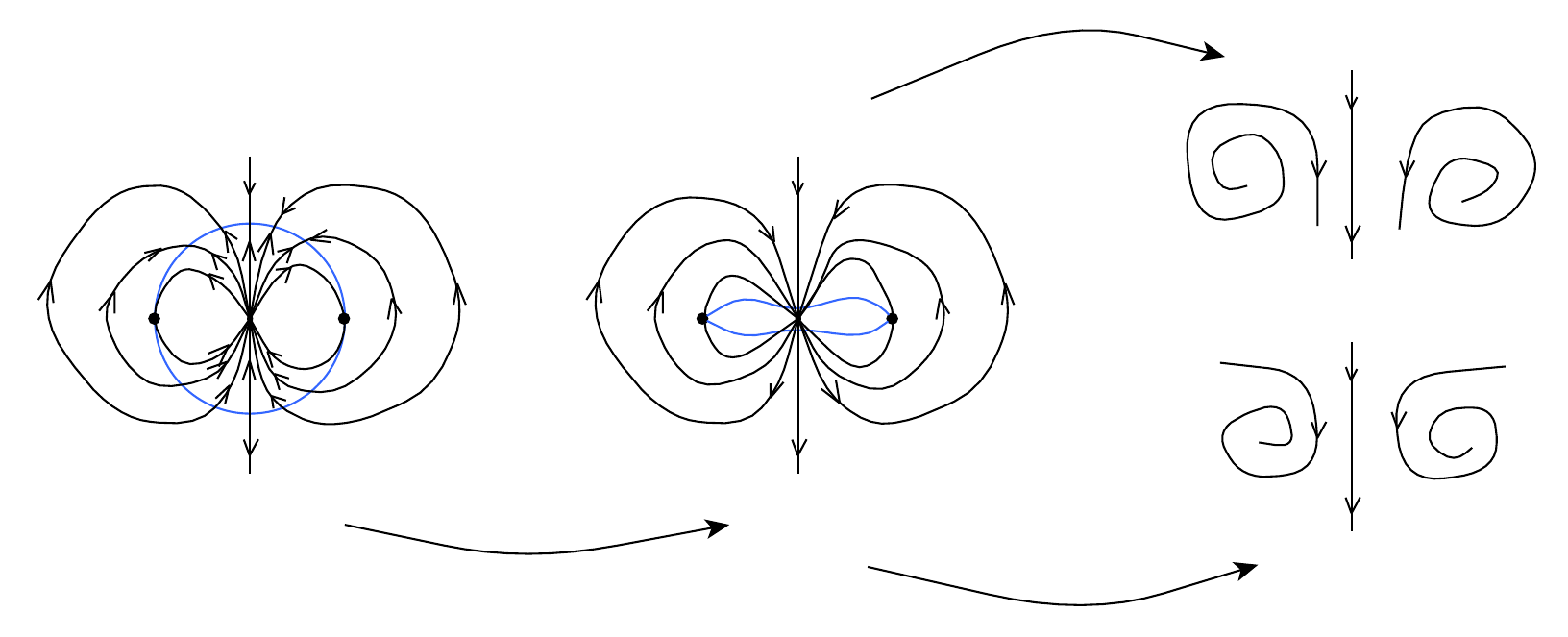}
	} 
	\vskip0cm
	\caption{The disappearance of degeneracy for $E_2$.}
	\label{F:2_zero_c}  
\end{figure} 
	
\begin{figure}[!htb]
	\centerline{				
		\includegraphics[width=.5\textwidth]{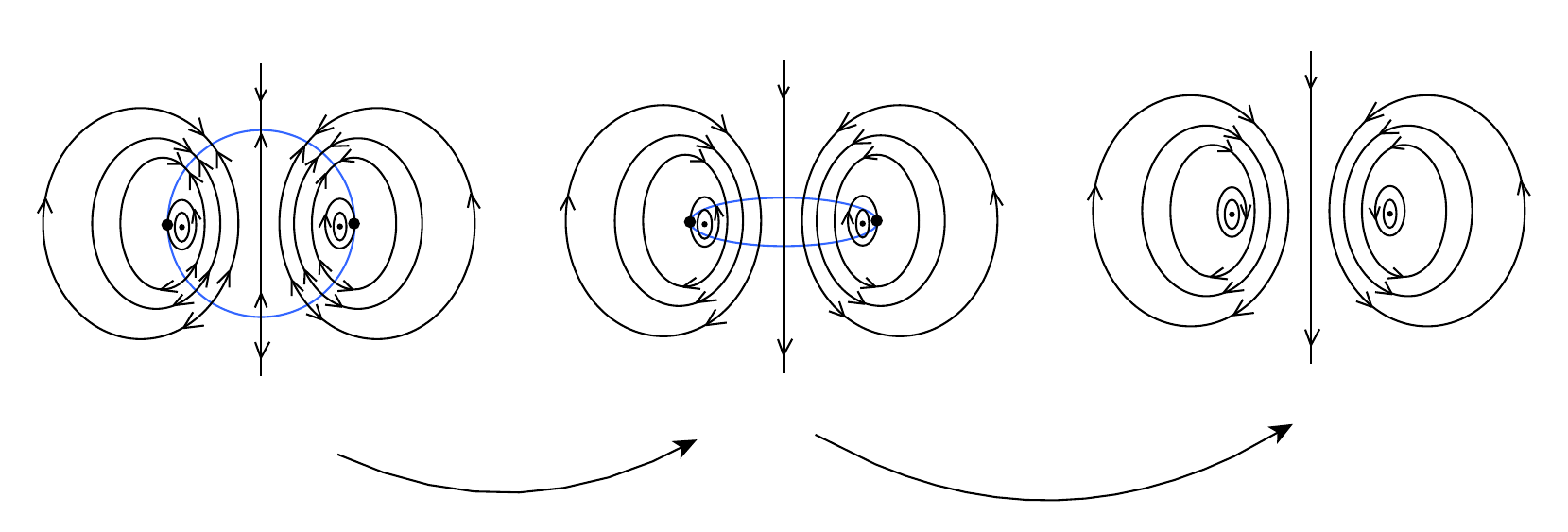}
	} 
	\vskip0cm
	\caption{The disappearance of degeneracy for $\tilde E_2$.}
	\label{F:2_zero_d}  
\end{figure} 	
\END
\end{ex}

In general, a non-zero even number of zeros of $\la JE, \nabla f\ra$ on the ring gives rise to the collapse of the  degeneracy ring, by contracting the in- and out-flow pair-wise. Therefore, we have

\begin{proposition}\label{prop:robust}
The only degeneracy rings that are robust against deformations are those with ring index $\pm 1$. 
\end{proposition}
\begin{proof} Let $E$ be a vector field with a ring index different from $\pm 1$. Then, by Lemma \ref{lem:r1}, it has $\rind_\gamma(E)=0$ and $m_\gamma(JE)>0$. By Lemma \ref{lem:even}, $m:=m_\gamma(JE)$ is even. By properties of winding number, $E$ is homotopic to either $E_{\frac m2+1}$ or $E_{-\frac m2+1}$ (cf. Example \ref{ex:E_pm_n} for notations). In either case, the degeneracy ring can be deformed away by pairing up in- and out-flow inside the ring.
\end{proof}

\section{2D Sphere} 

Recall that the Euler characteristic of {$S^2$} can be realized by the sum of indices of all isolated zeros of any vector field on $S^2$  (having only isolated zeros), as stated by the Poincar\'e-Hopf theorem. Based on this result, flows on $S^2$ can be classified by the indices of their zeros and the sum of all indices is globally constrained by the topology of $S^2$.

We would like to establish a parallel of the Poincar\'e-Hopf theorem for degenerate flows on $S^2$ -which we assume to be orientable-, using degeneracy indices, which in the case of absence of degeneracy, reduces to the classical Poincar\'e-Hopf theorem.

Given a degenerate flow on $S^2$, we assume that the flow has only isolated zeros {or isolated limit cycles} besides degeneracy rings. We also assume that the degeneracy rings do not intersect any of these zeros {and limit cycles}, nor do they intersect each other. 

We will now define the {\it degeneracy index} for degeneracy rings, limit cycles and isolated zeros. This index includes, besides the winding number for isolated zeros $(S^1)$, a second attribute to account for the orientation of the flow $(\mathbb{Z}_1)$. A positive (resp. negative) coefficient for this second attribute indicates outgoing (resp. ingoing) flux. Thus, the $\mathbb{Z}_1$ component of source and sink isolated zeros have opposite signs, while a saddle carries zero $(\mathbb{Z}_1)$ component. Similarly, a creation and annihilation degeneracy surfaces have $(\mathbb{Z}_1)$ component of opposite signs.

This notation is borrowed from the $S^1$-equivariant degree, where zero orbits of an $S^1$-equivariant map are labeled by their symmetries (orbit types) {\cite{Ruan}}. The idea is that a zero orbit of an $S^1$-equivariant map is $1$-dimensional if it has a discrete symmetry $(\bz_m)$ for some finite subgroup $\bz_m\subset S^1$; or it is $0$-dimensional if it has the full  continuous symmetry $(S^1)$.

\subsection{Degeneracy Rings}\label{subsec:dr}
In Subsection \ref{subsec:def_ri}, we introduced an index for degeneracy rings, which takes value from $\pm 1$ for rotating fields $E$, otherwise it is equal to $0$. Also, as it has been shown in Subsection \ref{subsec:robust}, degeneracy rings enclosing a simply connected region are robust if and only if they have non-zero ring index.

Thus, we assign for the creation (resp. annihilation) degeneracy rings zero winding number and $+1$ (resp. $-1$) to its $\mathbb{Z}_1$ attribute. More precisely, a ring $\gamma$ is called an {\it annihilation} ring, if $\rind_{\gamma}(E)=-1$; it is called a {\it creation} ring, if $\rind_{\gamma}(E)=1$. Define {\it degeneracy index of $\gamma$} by
\begin{align}
\Ind: \q \text{annihilation $\gamma$}\q & \mapsto\q -(\bz_1)
\label{eq:di_ring}\\
\text{creation $\gamma$}\q & \mapsto\q \;\;\; (\bz_1).\notag
\end{align}
\\ 
See Figure \ref{fig:z1_pm} for the degenerate flows it implies.

\begin{figure}[!htb]
	\centerline{
		\includegraphics[width=.4\textwidth]{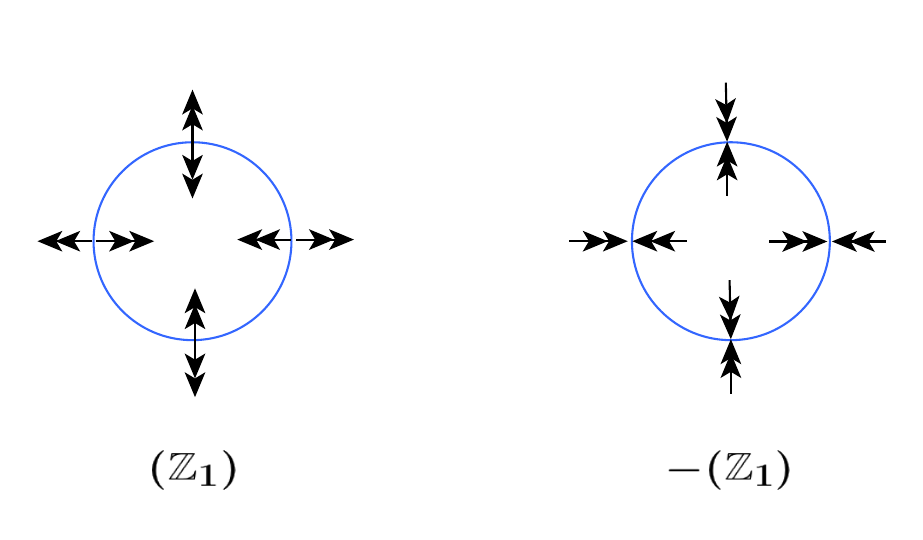}
	} 
\caption{Index of degeneracy rings: (left) $(\bz_1)$ for creation ring; (right) $-(\bz_1)$ for annihilation ring. The sign is determined by the ring index.}
\label{fig:z1_pm}
\end{figure}

\begin{rmk}\rm\label{rmk:ss_diff} Note that a creation ring can naturally enclose a sink inside by extending all the in-flow arrows. Similarly, an annihilation ring can enclose a source inside. However, if a creation ring were to enclose a source inside, there needs to be some additional structure mediating the two {such as a limit cycle}.
\END 
\end{rmk}

\subsection{Limit Cycles}
For isolated limit cycles, we adopt the usual $S^1$-equivariant degree and define the {\it degeneracy index} of a limit cycle by  (cf. \cite{Ruan}) 
\begin{align}
\Ind: \q \text{attracting cycle}\q&\mapsto\q -(\bz_1)
\label{eq:cycle_ind}\\
\text{repelling cycle}\q&\mapsto\q \;\;\; (\bz_1).\notag
\end{align}

\subsection{Isolated Zeros}
From the discussion of Subsection \ref{subsec:robust}, degeneracy rings that  enclose isolated zeros other than sources or sinks cannot be robust. Thus, it is sufficient to define degeneracy index for sources and sinks.

The following observation provides a foundation for the definition.
\begin{lem}\label{lem:ss_ah} {In absence of limit cycles,} a sink is admissibly homotopic to an annihilation ring enclosing a source; a source is admissibly homotopic to a creation ring enclosing a sink (cf. Figures \ref{fig:ss_1} and \ref{fig:ss_2}).
\end{lem}	

\begin{proof} To prove the first part, suppose that we have an annihilation ring enclosing a source. Example \ref{ex:f_2}(a) gives a vector field $E^-$ together with $f=x_1^2+x_2^2-1$ that carries such a degenerate flow. By Proposition \ref{prop:ring_ind_2}(i), every such degenerate flow is indeed admissibly homotopic to $E^-$. Define a homotopy $H(s,x_1,x_2)=x_1^2+x_2^2-s$ for $s\in [-1,1]$ of $f$. The zero set of $H_s:=H(s,\cdot)$ forms a ring of radius $\sqrt{s}$ for $s>0$, which shrinks to a point for $s=0$ and disappears for $s<0$. It gives rise to an admissible homotopy on rings $\gamma_s:=H_s^{-1}(0)$ of degeneracy, since $m_{\gamma(s)}(JE^-)=0$ for all $s\in [-1,1]$ {and there are no limit cycles by assumption}. The deformation of flow is described by the following system
	\[(x_1^2+x_2^2-s)\left(\begin{array}{c} \dot x_1\\\dot x_2 \end{array}\right)=JE^-=\left(\begin{array}{c} -x_1\\-x_2 \end{array}\right),\q s\in[-1,1],\]
which has a sink for $s=-1$ and an annihilation ring enclosing a source for $s=1$. See Figure \ref{fig:ss_1}.
\begin{figure}[!htb]
	\centerline{
		\includegraphics[width=.45\textwidth]{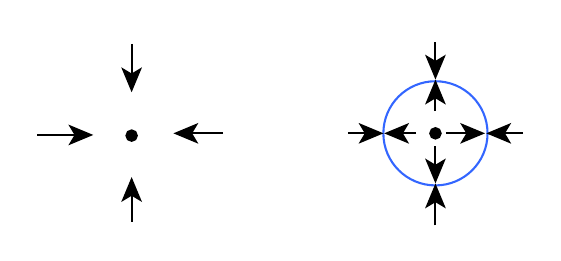}
	} 

	\caption{A sink (left) and an annihilation ring enclosing a source (right) are admissibly homotopic.}
	\label{fig:ss_1}
\end{figure}  

The other case is analogous and can be proved using $E^+$ from Example \ref{ex:f_2}(b). See Figure \ref{fig:ss_2} for the dynamics.
\begin{figure}[!htb]
\vskip-.1cm
\centerline{
\includegraphics[width=.45\textwidth]{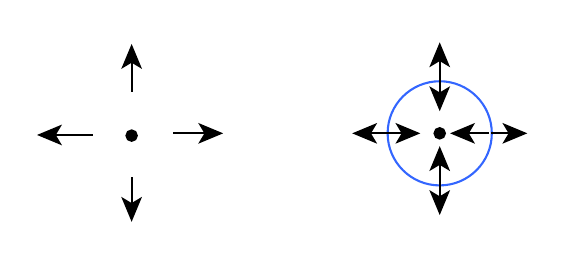}
} 

\caption{A source (left) and a creation ring enclosing a sink (right) are admissibly homotopic.}
	\label{fig:ss_2}
\end{figure}  
\end{proof}

In other words, a sink can be viewed as an annihilation ring collapsed to a point. Similarly, a source can be viewed as a creation ring collapsed to a point.

It follows from Lemma \ref{lem:ss_ah} that degeneracy indices for a source and for a sink differ by a degeneracy ring with an appropriate sign. For example, if we associate $(S^1)+a(\bz_1)$ to a source with an orientation coefficient $a\in \mathbb{R}$, then the index for a sink is equal to $(S^1)+(a-1)(\bz_1)$, since it is the index of a source subtracted by the index of a creation ring which is $(\bz_1)$. For symmetry consideration, we choose $a=\frac 12$ so that by reversing the flow orientation around sources and sinks, the indices are exchanged by taking an opposite sign in front of $\bz_1$. 

Define the {\it degeneracy index} of a source and a sink by the following.
\begin{align}
\Ind: \q \text{source}\q&\mapsto\q (S^1)+\frac 12(\bz_1)\label{eq:di_zero}\\
\text{sink}\q&\mapsto\q (S^1)-\frac 12(\bz_1).\notag
\end{align}
For an isolated zero that is neither a source nor a sink, the degeneracy index will have zero orientation coefficients and take the form of $d(S^1)$ for $d\in \bz$ being the {winding number}. For example, the degeneracy index of a saddle is $-(S^1)$.

\subsection{Generalized  Poincar\'e-Hopf Theorem} 

Based on the degeneracy index defined by (\ref{eq:di_ring})  and (\ref{eq:di_zero}), we have the following extension of Poincar\'e-Hopf theorem \cite{Milnor}. 

\begin{proposition}\label{prop:ind} 
For every flow of (\ref{eq:sym2d_gen_JE}) generated by $f$ and $JE$ on $S^2$, we have
\begin{equation}\rc{\label{eq:gt}}
\sum \Ind = 2(S^1),
\end{equation}
where {the sum} is taken over all isolated zeros and degeneracy rings of the flow.
\end{proposition}

\begin{proof} A {special} case that we consider first is when the flow does {\it not} possess any degeneracy rings, that is, a case of regular flow. Then, such flow is homotopic to a flow that contains only one source and one sink, for which we have
\[\sum \Ind = (S^1)+\frac 12 (\bz_1)+(S^1)-\frac 12 (\bz_1)=2(S^1).\]
\vs
	Otherwise, consider a flow that has degeneracy rings. Assume, without loss of generality, that they are robust and do not intersect each other on $S^2$. It follows that they form parallel rings on $S^2$ and by Proposition \ref{prop:robust}, they have ring index $\pm 1$, {corresponding} to degeneracy index $\pm (\bz_1)$.  {In case that there are limit cycles, they also form parallel rings with corresponding degeneracy index $\pm (\bz_1)$.}
\vs
 
Let $\gamma$ be a degeneracy ring that {encloses} no other degeneracy rings, that is the first or the last parallel {degeneracy} ring on $S^2$. {If $\gamma$ encloses no limit cycles, then} depending on its degeneracy index, it contains either a source (for $-(\bz_1)$-ring) or a sink (for $(\bz_1)$-ring). By Lemma \ref{lem:ss_ah}, one can then deform the ring to the contained fixed point by an admissible homotopy, which results in a sink or a source, in the respective two cases. Notice that throughout the deformation, the sum of the degeneracy index {does not change}. {Otherwise, if $\gamma$ encloses some limit cycles, one can repeat the above procedure for the limit cycles which can be deformed to a sink or source, depending on their stability types.}
\vs
Repeat this process for the {remaining} degeneracy rings {and limit cycles}, deforming them one after another to their contained fixed points by admissible homotopies, which every time makes a ring disappear by {switching between a source and a sink}. At the end, one obtains a flow without degeneracy rings, {whose sum of degeneracy indices equals} $2(S^1)$. Since the deformation procedure is achieved by {admissible} homotopies, the sum of degeneracy index {is preserved}. Thus, the sum of degeneracy index for the original flow is also equal to $2(S^1)$.
\end{proof}

\subsection{Interpretation of  $2(S^1)$}

The  sum $2(S^1)$ of degeneracy indices in Proposition \ref{prop:ind} describes the global topology of $S^2$ in the following way. Recall that the sum of ordinary indices  amounts to the Euler characteristic of $S^2$, which is given by an alternating sum using cells in a cellular decomposition of $S^2$. 

In a similar way, one can use degeneracy index to distinguish different cells for the construction of $S^2$. See Figure \ref{fig:gt2}, where indices for $2$D cells are defined by their enclosed singularities (cf. Lemma \ref{lem:ss_ah}). 
\begin{figure}[!htb] 
	\centerline{
		\includegraphics[width=.4\textwidth]{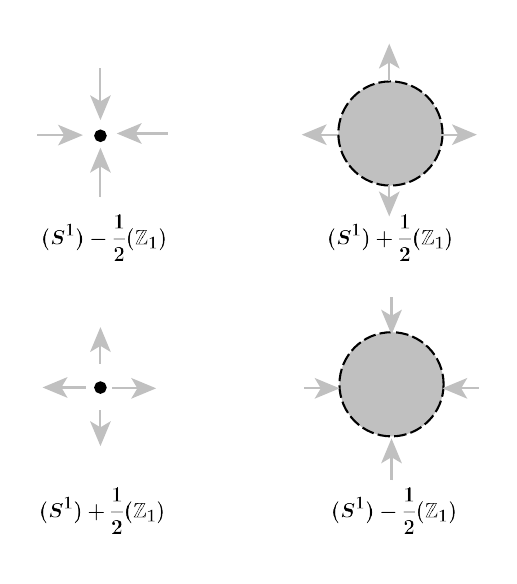} 
	} 
				\vskip.2cm
	\caption{Cells of dimension $0$ and $2$ with their indices. The upper resp. lower row represents a cellular decomposition of $S^2$ (cf. (\ref{eq:dc_sk})-(\ref{eq:dc_sc})).}
	\label{fig:gt2} 
\end{figure}  

For example, consider a (cellular) decomposition of $S^2$ given by
\begin{equation}\label{eq:dc_sk}
S^2=S^2\setminus\{\text{sink}\}\sqcup \{\text{sink}\},
\end{equation}
which is a combination of the two cells in the upper row of Figure \ref{fig:gt2}.  Alternatively, one can consider the decomposition of $S^2$ given by
\begin{equation}\label{eq:dc_sc}
S^2=S^2\setminus\{\text{source}\}\sqcup \{\text{source}\},
\end{equation}
which corresponds to the combination of the two  cells in the lower row of Figure \ref{fig:gt2}. 
Other decompositions involving $1$D cell are also possible. Consider the decomposition of $S^2$ by two $2$D cells and one $1$D (closed) cell
\begin{equation}\label{eq:dc_sn1}
S^2= B_N^+\sqcup B_S^+\sqcup S^{-},
\end{equation}
where $B_N^+$ resp. $B_S^+$ denotes the open northern resp. southern hemisphere of $S^2$ with {\it outgoing} flow on the boundary and $S^{-}$ denotes the equator of {\it incoming} flow from both sides (cf. Figure \ref{fig:gt3}, the upper row). 
\begin{figure}[!htb]
\centerline{
\includegraphics[width=.45\textwidth]{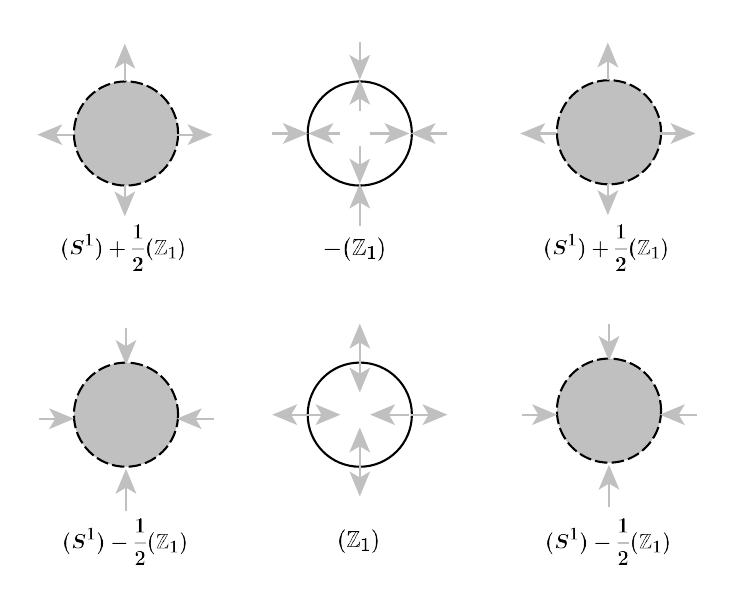} }
\vskip.2cm
\caption{Cells of dimension $1$ and $2$ with their indices.  The upper resp. lower row represents a cellular decomposition of $S^2$ (cf. (\ref{eq:dc_sn1})-(\ref{eq:dc_sn2})).}
\label{fig:gt3} 
\end{figure}   
Alternatively, 
\begin{equation}\label{eq:dc_sn2}
S^2= B_N^-\sqcup B_S^-\sqcup S^{+},
\end{equation}
is another decomposition, where $B_N^-$ resp. $B_S^-$ denotes the open northern resp. southern hemisphere of $S^2$ with {\it incoming} flow on the boundary and $S^{+}$ denotes the equator of {\it outgoing} flow from both sides (cf. Figure \ref{fig:gt3}, the lower row).  Notice that in all the decompositions (\ref{eq:dc_sk})-(\ref{eq:dc_sn2}), the sum of indices of cells is equal to $2(S^1)$.

\section{Conclusion}   
We have introduced a degeneracy index for $2$D flows and shown that it can be used to extend the Poincar\'e-Hopf Theorem for degenerate flows on $S^2$. But it remains unclear how it can be extended to other $2$D compact surfaces such as $g$-surfaces in a straighrforward way. The analysis of the robustness of degeneracy rings benefits from the simplicity of the topology of $S^2$ (cf. Subsection \ref{subsec:robust}). Also, the current discussion in $2$D is motivated by physics applications, but it is also interesting to explore higher dimensions. Careful readers may have noticed we have used the symbols $(S^1)$ and $(\bz_1)$ to distinguish zeros and degeneracy rings on $S^2$. It is not clear, however, how this choice can be made in a canonical way and how to systematically define a degeneracy index for degenerate flows in  higher dimensional phase spaces.\\

\noi {\bf Acknowledgments}
This work has been partially funded by Fondecyt Grant No.1180368. HR thanks JZ for his invitation to visit CECs, during which the work has found its fast convergence. HR would also like to express her gratitude to Dr. Ivan Ovsyannikov for insightful discussions.

\vs

\end{document}